\documentclass[10pt, conference, compsocconf]{IEEEtran}
\usepackage{url}
\usepackage{amssymb,amsmath,amsfonts}
\usepackage{amsmath}
\usepackage{graphicx}
\usepackage{times}
\usepackage{algorithm}
\usepackage[algo2e, ruled, vlined]{algorithm2e}

\newcommand{\ignore}[1]{}
\newcommand{\notinproc}[1]{#1}
\newcommand{\onlyinproc}[1]{}

\newtheorem{thm}{Theorem}[section]

\newtheorem{lemma}[thm]{Lemma}

\newtheorem{corollary}[thm]{ Corollary}

\newcommand{\ellth}{\mathop{\textstyle {\ell}^{\text{th}}}}
\newcommand{\ith}{\text{i}^{\text{th}}}
\newcommand{\Inf}{\mathop{\rm Inf}}

\newcommand{\increase}[2]{\ensuremath{#1 \, {\overset{+}{\gets}}\, #2}}
\newcommand{\decrease}[2]{\ensuremath{#1 \, {\overset{-}{\gets}}\, #2}}
\SetKwFunction{Append}{append}
\SetKwArray{SeedList}{seedlist}
\SetKwFunction{MargGain}{MargGain}
\SetKwArray{cdelta}{$\delta$}
\SetKwArray{digest}{uDigest}
\SetKwArray{Index}{index}
\SetKwArray{HM}{HM}
\SetKwArray{ML}{ML}
\SetKwArray{est}{Est}
\SetKwFunction{MoveUp}{MoveUp} 
\SetKwFunction{MoveDown}{MoveDown}
\SetKwArray{Qgrand}{Qelements} 

\newcommand{\revSO}{\text{{\sc RevSortedAccess}}}
\newcommand{\forSO}{\text{{\sc ForwardSearch}}}

\title{Greedy Maximization Framework for  \\ Graph-based Influence Functions}
\author{
Edith Cohen \\
{\large Google Research, USA}\\
{\large Tel Aviv University, Israel}\\
{\large edith@cohenwang.com}
}

\begin{document}
\maketitle
\thispagestyle{plain}
\pagestyle{plain}

\begin{abstract}
The study of graph-based submodular maximization problems 
was initiated in a seminal work of Kempe, Kleinberg,
  and Tardos (2003):  An 
{\em influence} function of
subsets of nodes is defined by the graph structure and the aim
is to find subsets of seed nodes with (approximately) optimal tradeoff of
size and influence. Applications include viral 
marketing, monitoring, and active learning of node labels.
This powerful formulation was studied for
 (generalized) {\em coverage} functions, where the influence of a seed set
on a node is the maximum utility of a seed item to the node,
and for pairwise {\em utility} based on reachability,
distances, or reverse ranks.    

 We define a rich class of influence functions which unifies and extends
 previous work beyond coverage functions and specific utility functions.
 We present a meta-algorithm for  approximate greedy maximization with
strong approximation quality guarantees and worst-case near-linear computation for all functions
in our class. Our meta-algorithm generalizes a recent design by 
Cohen et al (2014) that was specific for distance-based coverage functions.
\end{abstract}

\section{Introduction}
Submodular maximization problems are extensively studied and used in
many application domains.  The aim is to compute a subset $S$ of {\em
  items} of a certain size that maximizes some submodular and monotone
{\em influence} (valuation) function $\Inf(S)$. 

 Even for the special case of coverage problems, 
the problem of computing, for parameter $s$, a set $S$ of $s$ items 
 with maximum $\Inf(S)$ is NP hard and also hard to approximate.  
A simple and practical algorithm is {\em Greedy}, which sequentially selects the item $$i = \arg\max_h \Inf(h \mid S)$$
 with maximum {\em marginal influence}  $$\Inf(i \mid S) = \Inf(S \cup
 \{i\}) - \Inf(S)\ .$$
We refer to the computed 
permutation of items as the {\em greedy sequence} and to the process
as {\em greedy maximization}.

A classic result of Nemhauser et al \cite{submodularGreedy:1978} shows 
that for submodular monotone functions,  any
$s$ prefix of the greedy sequence of items has influence that is at
least $1-(1-1/s)^s \geq 1-1/e$ of the maximum possible. 
Feige~\cite{feige98} showed that this is the best
worst-case approximation ratio we can hope for by a polynomial time
algorithm.  Hence the greedy sequence approximates the full Pareto
front of seed set size versus influence.
A very useful relaxation
is {\em approximate greedy}
which selects in each step an approximate maximizer that has at least
$(1-\epsilon)$ of the maximum marginal influence.  Approximate greedy
often scales up the computation while lowering
the approximation ratio guarantee by at most $O(\epsilon)$.

  We are interested  here in influence functions that are expressed in
  terms of {\em utility} values $u_{ij}$ between our items and {\em elements}.
The influence of the seed set $S$  can  then be defined as 
 the sum over elements $j$  of the maximum utility of a seed item to $j$:
\begin{equation} \label{infmax:eq}
\Inf(S) = \sum_j \max_{i\in S} u_{ij}\ .
\end{equation}
Coverage functions are the special case where
for any elements $j$,  there is $c_j>0$ such that $u_{ij}\in \{0,
c_j\}$.
A natural extension replaces the maximum in \eqref{infmax:eq} 
with another submodular aggregation
function $F$ applied to the multiset $\{u_{ij} \mid i\in S\}$.

Graphs are a common model of representing relations between entities\onlyinproc{.}
\notinproc{:
Edges represent stronger affinity between their end points but more
generally, affinity can be derived from the ensemble of paths  connecting nodes.
A sparse graph can therefore represent dense and intiricate affinity
relations between all pairs of nodes.
With graph-based influence, items and elements are nodes in a graph
and utility values are affinity. Influence is derived from utility using max
aggregation \eqref{infmax:eq}.   Seed sets with high influence and
small size optimize coverage, diversity, and information and have low redundancy.
Among the many applications are selecting
anchors/hubs/coresets for monitoring/analysing/sparsifying the 
network,  active learning candidates
in semi-supervised learning context, or viral marketing in social
networks.  
Graph-based influence is rooted in classic graph formulation of the
set (maximum) cover problem and notions of centrality (influence of a single
vertex) in social networks
\cite{Bavelas:HumanOrg1948,Freeman:sn1979}.  It 
was popularized and extended in a seminal 
 work of Kempe, Kleinberg, and Tardos \cite{KKT:KDD2003}.}
\onlyinproc{The notion of graph-based influence was introduced in a seminal 
 work of Kempe, Kleinberg, and Tardos \cite{KKT:KDD2003}.
Items and elements are nodes in a graph and the utility function is
derived from the graph structure. Influence is computed from utility using max
aggregation \eqref{infmax:eq}.
Applications of high-influence seed sets include using them as 
anchors or hubs for monitoring or sparsifying the 
network,  active learning 
in semi-supervised learning context, or viral marketing in social networks. }

There are several natural ways to derive utility values from graph 
structure (See example
in Figure~\ref{example:fig}).  The hugely popular Independent Cascade (IC) model of
\cite{KKT:KDD2003} uses $u_{ij}=1$ if $j$ is reachable from $i$ and $0$
otherwise.  Gomez-Rodriguez et al proposed distance-based
(``continuous time'') utility, where  edge lengths model propagation
times \cite{Gomez-RodriguezBS:ICML2011,DSGZ:nips2013}:  For a
threshold parameter $T$, $u_{ij}=1$ if and only if $d_{ij}\leq T$,
where $d_{ij}$ is the shortest path distance.  More generally,
\cite{timedinfluence:2015}  considered utility
that smoothly decreases with distance.
Smooth distance-based utility was studied in social network analysis \onlyinproc{(centrality
is the influence of a single node)}
\cite{Bavelas:HumanOrg1948,Sabidussi:psychometrika1966,Freeman:sn1979},
economic models~\cite{BlochJackson:2007},
and general network analysis \cite{CoKa:jcss07}. 
Reverse-rank utility is based on the order induced by distances 
instead of the magnitudes and by doing that, it ``factors out'' the
effect of varying density.  The special case of ``reverse nearest
neighbors'' influence, where
$u_{ij}=1$ if $i$ is 
the nearest neighbor of $j$, was formalized by Korn and 
Muthukrishnan \cite{KornMuthu:sigmod2000}: The influence of a seed set is 
the number of reverse nearest neighbors it has.  Buchnik and Cohen 
\cite{reverseranks:sigmetrics2016} generalized it to
higher order reverse near neighbors and smooth decrease with reverse
rank.

 A powerful extension is to use a
randomized model that selects edges or edge lengths from the base
 graph \cite{KKT:KDD2003,CDFGGW:COSN2013,ACKP:KDD2013}.
 Utility  values are then defined as the expectation.  
\notinproc{This allows reachability or distance-based utility values to capture finer properties of the connecting paths ensemble:  In particular,
utility increases not only for shorter paths but also with more
disjoint paths.  We note here that popular
spectral kernels (e.g. effective resistances or hitting probabilities of random walks) 
\cite{Chung:Book97a} share these basic qualities.}
\onlyinproc{This allows our
 reachability or distance-based utility values to capture finer properties of the connecting paths ensemble:  In particular,
utility increases for shorter paths but also increases with more
disjoint paths.}
The IC model of
 \cite{KKT:KDD2003} selects edges from the base set with independent
 probabilities.  Distance-based models assign edge lengths drawn from
(typically exponential or Weibull)
 distributions \cite{Gomez-RodriguezBS:ICML2011,DSGZ:nips2013,CDFGGW:COSN2013,ACKP:KDD2013}.   Computationally, randomized models are handled by
 generating multiple sets of edges using Monte Carlo simulations and
averaging over simulations.

  Massive data sets pose significant scalability challenges.
The exact greedy algorithm is polynomial but we seek algorithms that
are near linear.  A fruitful research thread proposed many
different heuristics and algorithms for influence maximization
\cite{Leskovec:KDD2007,CWY:KDD2009,CELFpp:WWW2011,TXS:sigmod2014}.
Most of these algorithms can guarantee both near-linear computation
and approximation quality only for small seed set sizes.  The only
approach  that computes a full approximate greedy sequence using
near-linear computation is the SKIM (sketch
based influence maximization) algorithm
of Cohen et al \cite{binaryinfluence:CIKM2014}, which 
efficiently maintains samples of the marginal influence sets of nodes.
SKIM was initially
designed  for reachability-based influence and extended to
distance-based \cite{timedinfluence:2015} and reverse-rank influence
\cite{reverseranks:sigmetrics2016}.  

\begin{figure}
\center
\mbox{
{\scriptsize 
\begin{tabular}{l|ll}
\multicolumn{3}{c}{utility  with $\alpha(x)=1/x$} \\
distance &   $d_{BA}=2$ &$u_{BA}=1/2$ \\ & $d_{CA}=1$ &$u_{CA}=1$\\
       &$d_{DA}=5$ & $u_{DA}=1/5$\\
\hline
reverse- & $\pi_{AB}=3$& $u_{BA}=1/3$\\
rank      & $\pi_{DB}=4$& $u_{BD}=1/4$\\
\hline
survival- &  $\tau_{AB}=2$ &$u_{AB}=2$ \\
thresh &  $\tau_{DB}=1$ &$u_{DB}=1$
\end{tabular}}}
\raisebox{-0.4in}{\includegraphics[width=0.2\textwidth]{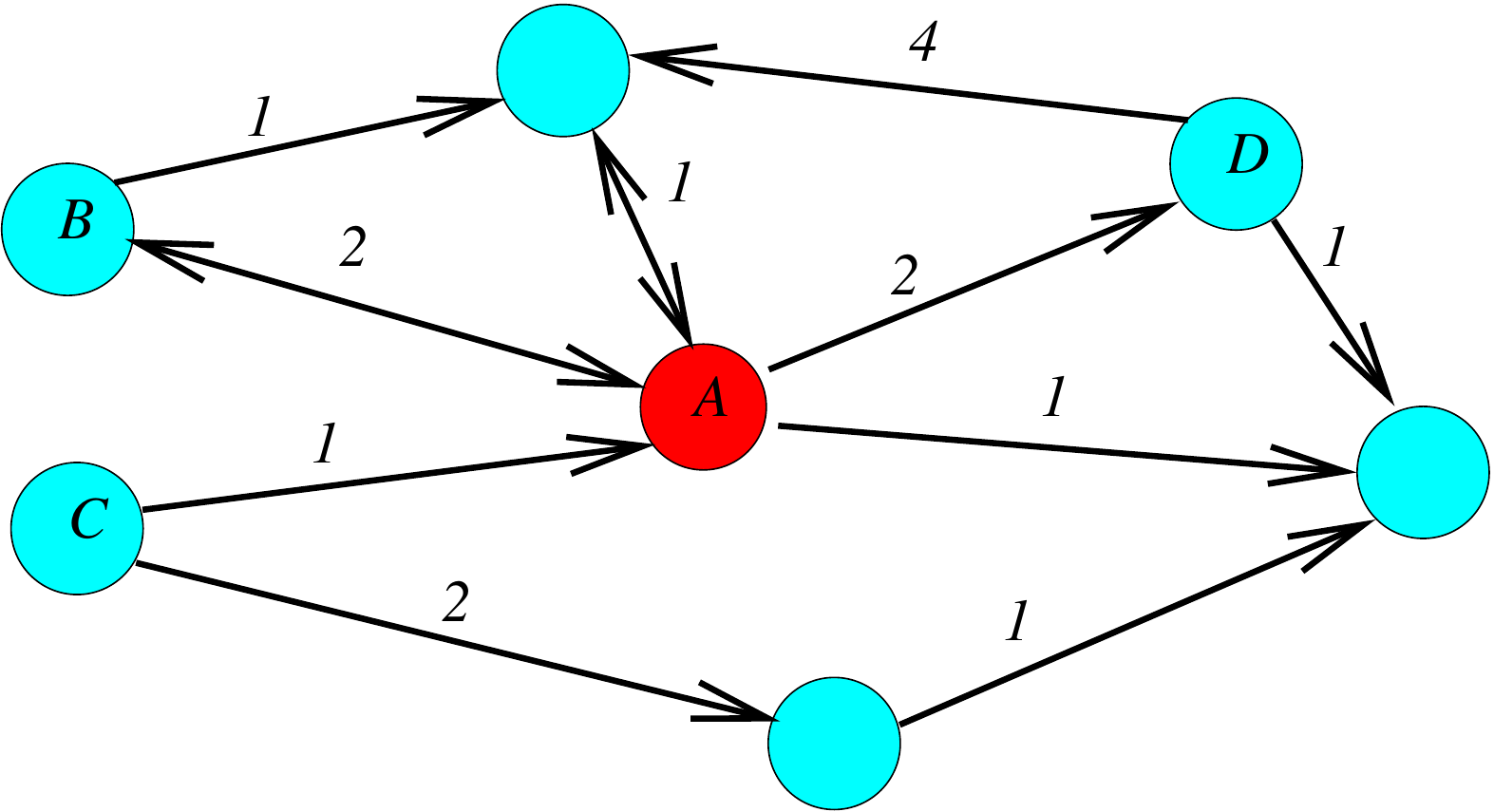} }
\caption{{\small Utility example for toy graph with edge lengths/lifetimes}
\label{example:fig}}
\end{figure}

\subsection*{Contribution and outline}
  Our main contribution here is presenting a general framework for
efficient  greedy influence maximization that generalizes and extends this
previous work.

Previous work focused on influence functions with
 {\em max aggregation} \eqref{infmax:eq},  where
the contribution of an element $j$ to 
the influence of the seed set $S$ is equal to the maximum utility
$\max_{i\in S} u_{ij}$ of a seed item.
With max aggregation, the value of a seed set of videos to a user is equal to
 that of their favourite video in the seed set.  Often,  however,
 elements derive additional value from other seed items.  
In our example, user $j$ may be 
 able to watch another video and thus values its second-favourite seed
 video $i$ at  $u_{ij}/2$.  
\notinproc{
In our toy graph example of Figure~\ref{example:fig} when considering
distance-based utility, the value of the seed set $\{B,C,D\}$ to
node $A$ is $1$ with max aggregation, and is equal to $u_{CA}$,  but
is $u_{CA}+u_{BA}/2 = 1.25$ when the second
favourite $B$  contributes half of $u_{BA}$.}
 In Section~\ref{extint:sec} we  define {\em submodular top-$\ell$}
aggregation  functions that depend on
the $\ell$ highest utility values of seed items to the element.
We also show that when the utility matrix $\{u_{ij}\}$ is provided explicitly,
an approximate greedy 
sequence can be computed in time that is 
near-linear in the number of nonzero  entries.

In a graph setting, explicit computation of $\{u_{ij}\}$ is
computationally prohibitive.
  SKIM and its extensions are efficient because dense utility values $\{u_{ij}\}$ are represented by a sparse
structure and SKIM performs the greedy maximization using computation
that is near-linear in the size of the structure.  SKIM  ultimately
obtains utility values for  a number of pairs that is near-linear
in the number of nodes.  In Section \ref{oracles:sec}
we present an abstraction of two 
access primitives  to the utility matrix that suffice for performing a 
``SKIM-like'' influence maximization computation:
{\em reverse
  sorted access} oracle and {\em forward search} oracle.  
In Section~\ref{metaskim:sec} we present the {\em SKIM meta-algorithm}  that performs
approximate greedy influence maximization for influence functions
with {\em submodular top-$\ell$} aggregation using near-linear computation and number
of oracle calls.

In Section~\ref{graphoracles:sec} we tie back our meta-algorithm to 
previous work by reviewing how the two access oracles
 are realized for reachability, distance, and reverse rank utilities.
 We also present oracles for {\em Survival threshold}
 utility \cite{semisupInf:2016}, which generalize 
 reachability utility\notinproc{ and is inspired by survivability
   analysis \cite{MillerSurvivalAnalysis:book}}:  For graph with edge weights interpreted as
{\em lifetime} values, the
survival threshold $\tau_{ij}$ is the maximum $t$ such that $j$ is
reachable from $i$ through edges with lifetime at least $t$. We use
$u_{ij} = \tau_{ij}$ (See example in Figure~\ref{example:fig}).

Our near-linear meta-algorithm for 
approximate greedy maximization is the first to apply 
to influence functions with (i)~previously-studied 
utility and aggregation function other than maximum,
(ii) smooth reverse-rank utility, and (iii) survival threshold utility.

\section{Influence functions} \label{extint:sec}
  We consider here influence functions of a particular form.
We have  $n$ {\em items},
{\em elements}, pairwise {\em utility} values $u_{ij} \geq 0$ of an 
item to an element, and an {\em aggregation function} $F$ that
is applied to a multiset of numbers.

We define the utility $$u_{Sj} = F(U_{Sj})$$ of a seed set $S$ of items to an
element $j$ as the aggregation function $F$ applied to the multiset
of pairwise values $$U_{S}(j) = \{u_{ij} \mid i\in S\}\ .$$
Finally, the {\em influence} of a seed set $S$ of items 
is defined as the sum over 
elements $j$ of the utility $u_{Sj}$ of $S$ to the element
$$\Inf(S) = \sum_j u_{Sj} = \sum_j F(U_{S}(j))\ .$$

The simplest and most common aggregation function is the maximum
$F(U_{Sj}) = \max_{i\in S} u_{ij}$.   
We define a natural class of more general aggregations that
are monotone 
 submodular functions  of the $\ell$ largest values in $U_{Sj}$. We start with 
a useful definition of a {\em domination} partial order on
multisets of positive numbers:
$$A \succeq B \iff \forall i, \ith(A) \geq \ith(B), $$
where $\ith(A)$ is the $i$th largest value in $A$ when $i\leq |A|$ and
is $0$ otherwise. 
For a parameter $\ell$, a function $F$ is {\em submodular top-$\ell$} if and only if:
\begin{eqnarray}
&& F(\emptyset) =   0 \\
&& \forall a>0,  F(\{a\}) = a \\
&& \forall A,\ F(A) = F(\text{top-}\ell(A)) \\
&& \forall A,B\   A \succeq B \implies  F(A) \geq F(B)\ . 
\end{eqnarray}
Some examples of {\em submodular top-$\ell$} functions are:
 max $F(A) = \max_{a\in A} a$,  sum of top-$\ell$ values $F(A)=\sum_{i=1}^\ell \ith(A)$, or a weighted sum $F(A)=\sum_{i=1}^\ell \frac{1}{i} \ith(A)$

\begin{lemma} \label{growinv:lemma}
If $F$ is {\em submodular top-$\ell$}  then
$$\forall A,B,C \, \,  A \succeq B \implies  F(A\cup C) \geq F(B\cup C)\ . $$
\end{lemma}
\begin{proof}
Note that for all $C$, $$A \succeq B \implies A\cup C \succeq B\cup C$$
\end{proof}

 We show that when the utility matrix $u_{ij}$ is provided explicitly, the {\em lazy approximate} greedy algorithm
 has a guaranteed approximation ratio of $1-1/e-\epsilon$ using computation that is
 near-linear in input sparsity (number of nonzero entries).

\subsection{Utility Digest} \label{digest:sec}
To efficiently compute marginal utility as seed items are added, we maintain
a {\em utility digest}, $\digest{j}$, for each element $j$. The digest is a
summary of $U_{Sj}$ with  internal implementation that depends on the aggregation 
function $F$. It always suffices to store the  $\ell$ largest values
in $U_{Sj}$, but a compact representation (e.g. histograms) suffices
for some $F$.  We will use the following operations:
{\onlyinproc{\small}
\begin{itemize}
\item
$\digest{j}.init$  
initializes an empty digest with threshold $0$.
\item
$\digest{j}.thresh \gets \inf_x F(U_{Sj}\cup\{x\}) >  F(U_{Sj})$
returns the threshold value $x$ that can increase utility of the seed set.
\item 
$\digest{j}.marg(x) \gets F(U_{Sj}\cup\{x\}) - F(U_{Sj})$
returns the marginal gain of adding a seed $i\not\in S$ with utility $x=u_{ij}$
\item 
$\digest{j}.{\text AddMarg}(y,x) \gets F(U_{Sj}\cup\{y,x\}) -
F(U_{Sj}\cup\{y\})$
For two items $h,i\not\in S$ with utility $y=u_{hj}$ and $x=u_{ij}$,
the marginal gain of adding $i$ if $h$ is already added to $S$.
\item 
$\digest{j}.val \gets F(U_{Sj})$
\item 
$\digest{j}.update(x)$:  Compute a digest of $U_{Sj} \cup\{x\}$ given $x$
and digest of $U_{Sj}$. Updating digest with a new seed item $i$ with utility $u_{ij}=x$.  
\end{itemize}
}
\subsection{Approximate lazy greedy} 
The algorithm maintains all items $i$  in a max heap with
 priorities equal to their marginal influence at the time of
 insertion.  The initial priority of $i$ is $\Inf(\{i\}) = \sum_j u_{ij}$.  
We iterate the following until the heap is empty.  We
 pop the item $i$ at the top of the heap and
 compute its exact marginal influence
$$\Inf(i \mid S ) = \sum_j \digest{j}.marg(u_{ij})\ .$$
 If  $\Inf(i \mid S )$ is at least
 $(1-\epsilon)$ of the current priority of item $i$, it is added to the
 seed set $S$ and we update $\forall j,\ \digest{j}.update(u_{ij})$.
Otherwise, if $\Inf(i \mid S) > \max_h\Inf(\{h\})/n^2$, item $i$ is placed back in the heap with current priority
 equal to $\Inf(i \mid S )$.

Since this is an approximate greedy sequence,
we have a guaranteed approximation ratio 
of $1-(1-1/s)^s - \epsilon$ for each $s$ prefix of the sequence.  The
computation is near linear when $\ell$ is small:
\begin{lemma}
Approximate lazy greedy uses $$O(m \ell \epsilon^{-1} \log n +
\epsilon^{-1} n \log^2 n)$$
computation, where $m$ is 
the number of nonzeros in $u_{ij}$.  
\end{lemma}
\begin{proof}
  Each marginal influence computation for item $i$  amounts to $2m_i$
 digest operations, typically $O(\ell)$ each, where
$m_i$ is the number of elements $j$ with $u_{ij}>0$.  Thus it is
$O(\ell m_i)$.

Each time an item is placed back on heap, 
its marginal influence decreased by at least a factor of $(1-\epsilon)$ from
the value it had when previously placed on the heap.  Hence, this can
happen $\epsilon^{-1} \log n$ times.  Therefore, the total computation
for item $i$ is $O(\epsilon^{-1} \ell m_i \log n )$ and $O(\epsilon^{-1}\log^2 n)$ for heap operations.  The claim follows by summing over items $i$.
\end{proof}

In the sequel we address settings where $u$ is dense ($m$ is much
larger than the number of items and elements)
or expensive to compute and show how the maximization can be performed
by only accessing ``relevant'' entries.

\section{Oracle access to utility values} \label{oracles:sec}

We define two access oracles to the utility matrix $\{u_{ij}\}$ that
allow us to perform approximate greedy maximization
while only retrieving a fraction of entries:
{\em Reverse sorted access} from elements, and {\em forward search} from items.


\subsection{Reverse sorted access}
The reverse sorted access oracle $\revSO$ for element $j$ returns
items and utility value pairs $(i,u_{ij})$ in non-increasing order of $u_{ij}$.
It supports the following operations:
{\small
\begin{itemize}
\item
$\revSO[j].init$: Initialize reverse sorted access from $j$
\item
$\revSO[j].top$: Return  $u_{ij}$ of next item without popping it.
\item
$(i,u_{ij}) \gets \revSO[j].pop$: pop return the next item in the sorted
order of $u_{ij}$.
\item
$\revSO[j].delete$:  Delete the data structure.
\end{itemize}
}

We use this oracle as the seed set $S$ grows. We show that for all
influence functions in our class and all
seed sets $S$, the marginal utility order  is the same as the utility order.
\begin{corollary}
$$\forall S,\ u_{ij} \leq u_{hj} \implies (u \mid S)_{ij} \leq (u \mid 
S)_{hj}\ ,$$
where $$(u \mid S)_{ij} \equiv u_{S\cup\{i\},j} - u_{S,j}\ .$$
\end{corollary}
\begin{proof}
We apply Lemma~\ref{growinv:lemma} with $C=U_{Sj}$, $A=\{u_{hj}\}$ and $B=\{u_{ij}\}$.
\end{proof}

\subsection{Forward search}
The forward search oracle $\forSO$ for item $i\not\in S$ returns
all elements $j$ for which 
$(u\mid  S)_{ij} > 0$ 
along with the value $u_{ij}$.
The oracle  supports initialization (with  respect to the seed set
$S$) and retrieving element utility pairs:
\begin{eqnarray*}
&&
\forSO[i].init(S)\\
&&
(j,u_{ij}) \gets \forSO[i].pop\ .
\end{eqnarray*}

The implementation of forward search is subtle.
Our influence functions allow that for some $i,j,S$
$u_{ij} > u_{ih}$ and $(u \mid S)_{ij} < (u \mid S)_{ih}$, so 
the implementation can not just use a sorted order of $j$ by
$u_{ij}$ that is oblivious to the current seed set $S$, but has to
adapt to $S$ to work efficiently.

\section{Graph-based utility} \label{graphoracles:sec}

  In graph-based settings, the input is represented by a graph or a set of graph {\em instances}  $\{G^{(h)}(V,E^{(h)})\}$ (obtained for example from Monte Carlo simulations of a randomized model).  All instances share the same set $V$ of nodes, which correspond to items. 
Our elements $j$ are  node instance pairs $(v(j),h(j))$.  The edges 
are directed and can have associated weights $w^{(h)}$, with 
interpretation that varies between definitions of utility.

 Reverse-sorted access is implemented by an appropriate graph search
 algorithm that is executed incrementally from a node in an instance 
 corresponding to element $j$ and returns nodes that correspond to items
 in decreasing $u_{ij}$ order.  Forward search 
is guided by an appropriate basic search algorithm by increasing
$u_{ij}$ with the additional property that the search tree partial order on elements $j$ preserves
marginal utility order:
If $h$ is a descendant of $j$ then for all $S$, $(u \mid
S)_{ih} \leq (u \mid S)_{ij}$.   Our forward search follows the basic
search tree and accesses digest structures of visited elements to compute
marginal utility value $\digest{j}.marg(u_{ij})$.  The search is
pruned when the marginal utility is $0$. Pruning is critical to
efficiency, since ultimately we perform forward searches from all
items.
The tree order property is critical for the pruning correctness -- so that it does not prevent us 
from reaching elements for which $i$ has positive marginal utility 
$(u \mid S)_{i j}>0$. 
\notinproc{The total number of nodes visited in a forward search are those adjacent to
nodes with positive marginal utility. }

We use the digest structures to guide
the forward search pruning.  Our SKIM meta-algorithm will use the forward
search to update the digest structures. Two useful subroutines are
{\sc MargGain}, which computes the marginal
 influence of an item given the current seed set, and {\sc
   AddSeed} which also updates digests of elements to reflect the
 addition of the new item to the seed set. 

\begin{function} \caption{{MargGain}($i$): Marginal influence of 
    $i\not\in S$ }
{\small 
\KwIn{item $i$}
\KwOut{$\Inf(i \mid S)$}
$M \gets 0$ \tcp*{sum of marginal contributions} 
$\forSO{i}.init(S)$ \tcp*{init forward search from $i$ with respect to 
  $S$} 
\While{$(j,u_{ij}) \gets \forSO{i}.next \not= \perp$} 
    {$\increase{M}{\digest{j}.marg(u_{ij})}$
}
\Return{$M$}
}
\end{function}

\begin{function}\caption{AddSeed($i$): Update digests of all 
 elements}
{\small 
\KwIn{item $i$}
$M \gets 0$ \tcp*{sum of marginal contributions} 
$\forSO[i].init(S)$ \tcp*{init forward search from $i$ with respect to 
  $S$} 
 \While{$(j,u_{ij}) \gets \forSO{i}.next \not= \perp$} 
     {$\increase{M}{\digest{j}.marg(u_{ij})}$\; $\digest{j}.update(u_{ij})$
 }
\Return{$M$}
}
\end{function}

\subsection{Distance-based utility}
 Distance-based utility is defined using a non-increasing function
 $\alpha$.  For item $i$ and element $j$, $u_{ij} =
 \alpha(d^{(h(j))}_{iv(j)})$, where $d^{(h)}_{ij}$ is the shortest
 path distance from $i$ to $j$ in instance $h$ with edge lengths $w^{(h)}$.

Reverse sorted access for element $j$ is implemented by Dijkstra's
 algorithm from $v(j)$ in the transposed (reversed edges) instance $h(j)$. Initialization places $v(j)$ in the heap and 
nodes are returned by increasing $d_{ij}$ (decreasing $u_{ij} = \alpha(d_{ij})$). 
When incoming edges are sorted by length, 
the computation is dominated by the number of traversed edges, which are edges adjacent to 
returned nodes.  

 Forward search from item $i$ is implemented by running a copy of
pruned  Dijkstra from node $i$ in each instance $h$.  
The computation is guided by the digest structure and pruned at nodes $j$ with 
$u_{i (h,j)} < \ellth(U_{S (h,j)})$.
\notinproc{Pruning correctness is
  established in  Lemma~\ref{prunecorrect:lemma}.}

\subsection{Reverse-rank utility}
 Reverse-rank utility is similarity defined as  $u_{ij} =
 \alpha(\pi^{(h(j))}_{v(j),i})$, where the Dijkstra rank
 $\pi^{(h)}_{j,i}$ is defined as the number of nodes that
 are at least as close to $j$ as $i$ is.
 Following Buchnik and Cohen
 \cite{reverseranks:sigmetrics2016}, we
  work with approximate ranks $\hat{\pi}$. The graph is preprocessed
  to obtain all-distances sketches, which are used to compute
  $\hat{\pi}_{ij}$ from $d_{ij}$. Approximation is necessary because search by exact reverse ranks
is provably as hard as all-pairs
  shortest paths computation \cite{reverseranks:sigmetrics2016}.

 The forward search implementation uses a pruned approximate reverse-rank 
search~\cite{reverseranks:sigmetrics2016}, which also traverses a shortest-path tree.
\notinproc{The correctness of pruning is established in Lemma~\ref{prunecorrect:lemma}.}
The reverse sorted access oracles uses an adaptation of Dijkstra on the transposed graphs that
is guided by approximate ranks.

\ignore{
 which was applied with uniform
  utility.  They used a Dijkstra-like {\em approximate} reverse-rank
  sorted access search, since sorted-access search by exact reverse-rank is hard
\cite{reverseranks:sigmetrics2016} (equivalent to
 all-pairs shortest path computation), a Dijkstra-like approximate
 reverse-rank sorted access search.  
The design uses All-Distances sketches to obtain estimates
$\hat{\pi}_{ij}$ from distances $d_{ij}$. We use utility $u_{ij} =
\alpha(\hat{\pi}^{(h(j))}_{iv(j)})$ based on approximate reverse rank.
  The reverse sorted access oracle also works with $\hat{\pi}$ instead
  of $\pi$ in order to   match the forward search.
}

\notinproc{
  Interestingly, our framework does not handle ``forward'' rank utility $u_{ij}=\alpha(\pi_{ji})$, 
as we are not aware of an efficient reverse sorted access implementation that is stable when seed nodes are added.}

\subsection{Reach and survival threshold utility}
 Reachability utilities are $u_{ij}=1$ if and only if there is a path from $i$ to $v(j)$ in instance $h(j)$. The more general survival threshold utilities are defined as $u_{ij} = \tau^{(h(j))}_{iv(j)}$ where $\tau^{(h)}_{ij}$ is the maximum $t$ such that there is a path from $i$ to $j$ in instance $h$ using edges with $w^{(h)}_e \geq t$.
  The reverse sorted access and the forward search oracles are similar
  to distance utility, with Dijkstra-like survival threshold search
  algorithm \cite{semisupInf:2016} replacing Dijkstra's algorithm.
\notinproc{
A survival-threshold search tree from source node $i$ has the property
that for any node $j$, the lifetime of all edges $e$ on the path from
$i$ to $j$ have $t_e \geq \tau_{ij}$.  By definition of $\tau_{ij}$, there must be at
least one edge on the path with $t_e = \tau_{ij}$.
}

\notinproc{
\subsection{Pruning correctness}
We establish correctness of the pruning performed by the forward 
search for distance-based, reverse-rank, and survival-threshold utility. 
\begin{lemma} \label{prunecorrect:lemma}
With distance-based utility, let
node $j$ be on the shortest path from $i$ to $r$ in instance $h$.
With (exact or approximate) reverse-rank utility, let $j$ be on a
shortest path from $r$ to $i$.   With survival-threshold utility, let 
$j$ be on a maximum survival threshold search tree path from
$i$ to $r$ in instance $h$.  Then,
$$u_{i (h,j)} < \ellth(U_{S (h,j)}) \implies (u \mid S)_{i (h,r)}=0\ .$$
\end{lemma}
\begin{proof}
We start with distance utility. 
For a node $a$, let $y_a$ be the $\ell$th smallest distance of $a$
from a node in $S$.  
A forward search from $i$ is pruned at a node $a$ if 
$d_{ia} \geq y_a$.  

  By definition, there are at least $\ell$ nodes $X\subset S$ such
  that $d_{jx} \leq y_j$.  From triangle inequality we have
  $$\forall x\in X,\ d_{xr}\leq d_{xj}+d_{jr} \leq y_j + d_{jr}.\ $$
Since there are at least $\ell$ nodes in $S$ with distance at most
$y_j + d_{jr}$ from $r$ we have
\begin{equation}\label{pc:eq3}  y_r \leq y_j + d_{jr}\ .\end{equation}
 Finally, if the search is pruned at $j$ then $d_{ij} \geq y_j$.  Combining
  with \eqref{pc:eq3} we obtain  $y_r \leq d_{ij}+d_{jr} = d_{ir}$.
Therefore, $(u \mid S)_{i (h,r)}=0\ .$

 For reverse-rank utility we apply a similar argument on a transposed
 graph.  We work with a shortest path from $r$ to $i$ that contains $j$. 
For a node $a$, let $y_a$ be the $\ell$th smallest rank $a$ has to a node in $S$. 
If the search is pruned at $j$ then 
$\pi_{ji} \geq y_j$ and in particular, there is a set $X\subset S$ of
$\ell$ nodes such that $\forall x\in X, d_{jx} \leq d_{ji}$.
 From the triangle inequality, for all $x\in X$,  $d_{rx} \leq d_{rj}+d_{jx}$.
From the shortest path property,  $d_{ri}=d_{rj} + d_{ji}$.
Combining, we obtain
\begin{equation}
\forall x\in X,\ d_{rx} \leq d_{rj}+d_{jx} \leq d_{rj}+ d_{ji} = d_{ri}\ .
\end{equation}
 This implies that $\pi_{ri} \geq y_r$ and therefore 
$(u \mid S)_{i (h,r)}=0$. 

We now consider survival-threshold utility.
For node $a$, let $y_a$ be the $\ell$th largest value in $\{\tau_{aj}
\mid a\in S\}$.
Since our path from $i$ to $r$ is a survival-threshold search path,
the minimum weight path edge between any two path nodes $h_1,h_2$ has weight
$\tau_{h_1 h_2}$.
 Let $e$ be the critical (minimum weight) edge  on the subpath from $j$
 to $r$. Then $\tau_{ir} = \min\{e, \tau_{ij}\}$.
 From definition, $y_r \geq \min\{y_j,e\}$.  If the search is pruned at $j$ then
$\tau_{ij} \leq y_j$. Combining we obtain 
$$\tau_{ir} = \min\{e, \tau_{ij}\}  \leq \min\{y_j,e \} \leq y_r$$ and therefore $(u \mid S)_{i (h,r)}=0$.
\end{proof}
}

\begin{algorithm2e}[h!]
\caption{SKIM meta-algorithm \label{metaskim:alg}}
{\scriptsize 
\DontPrintSemicolon 
\KwIn{IM problem: items, elements, $\revSO$ oracle for each element,
  $\forSO$ oracle for each item, function $F$}
\KwOut{Sequence of items}
\SetKwArray{Qgrand}{Qelements} 
\SetKwArray{Qcands}{Qitems} 
\SetKwArray{Qhml}{Qhml} 
\SetKwArray{Index}{index}
\SetKwArray{SeedList}{seedlist}
\SetKwArray{Rank}{rank}
\SetKwArray{cdelta}{$\delta$}
\SetKwArray{HM}{HM}
\SetKwArray{ML}{ML}
\SetKwFunction{NextSeed}{NextSeed} 
\SetKwFunction{MoveUp}{MoveUp} 
\SetKwFunction{MoveDown}{MoveDown} 
\SetKwFunction{UpdateReclassThresh}{UpdateReclassThresh} 
\SetKwArray{est}{Est}
\tcp{Initialization}
$\Rank \gets $ map elements to $\sim U[0,1]$ (or randomly permute elements 
$E$ and use $\Rank$ that maps $j$ to its position divided by $|E|$. 
);\\
\ForAll{elements $j$}{
$\Index{j} \leftarrow \perp$ \tcp*[l]{Reverse sample of $j$}
$\digest{j}.init$ \tcp*[l]{initialize digest of $j$}
$\revSO[j].init$ \tcp*[l]{Initialize reverse sorted access by $j$}
Insert  $j$ to $\Qgrand$ with priority $w_j \revSO[j].top/Rank[j]$
\tcp*[l]{place in $\Qgrand$ with priority $w_j \max_i 
  u_{ij}/\Rank{j}$}
}
\lForAll{items~$i$}{$\est.H[i] \leftarrow 0$; $\est.M[i]\gets 0$}
$\SeedList \leftarrow \perp$\tcp*{List of seeds \& marg.\ influences} 

$s\gets 0$;  $\tau  \gets\Qgrand.top/(2k)$;
$coverage \gets 0$ \tcp*{coverage of current seed set} 
\While{increase in coverage of last seed was at least $1/n^2$ of first seed}
{
\tcp{Build PPS samples of marginal influence sets until confidence in 
  next seed}
\While{$((x,\hat{I}_x) \gets \NextSeed()) =\perp$}
{
$\tau \gets \tau \lambda$;
$\MoveUp{}$ \tcp*{Update est. components} 
\ForAll{elements $j$  in $\Qgrand$ with priority $\geq \tau$}
{Remove $j$ from $\Qgrand$ \;
\While{ $(i,u_{ij}) \gets \revSO[j].top \not= \perp$}
{
\lIf(\tcp*[f]{utility below threshold}){$u_{ij} <  \digest{j}.thresh$}{
Terminate $\revSO[j]$; Break 
}
$c \gets \digest{j}.marg(u_{ij})$ 
\lIf(\tcp*[f]{Equivalently, $i\in S$}){$c ==0$}{$\revSO[j].pop$; Continue}

\If{$c/\Rank{j} < \tau$}{
  place $j$ with priority $c/\Rank{j}$ in $\Qgrand$; Break 
}
\Else(\tcp*[f]{$c/\Rank{j}\geq \tau$}){\revSO[j].pop \tcp*[l]{remove 
    $i$ from top of $\revSO[j]$} 
\Append $i$ to $\Index{j}$ \\
\lIf(\tcp*[f]{H entry}){$c \geq 
  \tau$}{$\increase{\est.H[i]}{c}$}\Else(\tcp*[f]{M 
    entry}){ $\increase{\est.M[i]}{1}$\;
  \If{$\HM{j}=\perp$}{$\HM{j}\gets |\Index{j}|$;  Insert $j$ with priority $c$ to $\Qhml$}
}
Update the priority of $i$ in $\Qcands$ to $\est.H[i]+\tau \est.M[i]$
\tcp*[l]{Can be lazy}
}
}
}
} 
\tcp{Process new seed item $i$}
$I_i \gets 0$ \tcp*{Exact marginal influence} 
Initialize $\forSO[i](S)$ \tcp*{Initialize forward search from $i$}
\While(\tcp*[f]{pop elements $j$ until $\perp$}){ $j \gets 
  \forSO[i].next \not= \perp$ }
{
\MoveDown{$j, u_{ij}$}\;  $\increase{I_i}{\digest{j}.marg(u_{ij})}$; $\digest{j}.update(u_{ij})$
}
$\increase{s}{1}$; $\increase{coverage}{I_i}$;
\SeedList.\Append{$i$,$\hat{I}_i$,$I_i$}
}
\Return{\SeedList}
} 
\end{algorithm2e}

\section{Greedy maximization} \label{metaskim:sec}
  We now present our SKIM meta-algorithm
  (Algorithm~\ref{metaskim:alg}) for approximate greedy maximization.
 The algorithm uses the oracle calls to access utility values.  It
maintains approximate marginal influence values for items (in a
lazy priority queue $\Qcands$), which
are computed from weighted samples.

We obtain weighted samples by assigning random {\em rank} values $r_j$ to elements. 
The $\revSO[j]$ oracles are used, together with a heap
structure $\Qgrand$ on elements, to return 
item element pairs by non-increasing order of $$\frac{w_j}{r_j} (u \mid
S)_{ij}\ .$$  The heap is prioritized by
$\frac{w_j}{r_j} (u \mid S)_{x(i)j}$, where $x(j)$ is the next item to be returned by
$\revSO[j]$.

The algorithm works with a  threshold value $\tau$
which decreases during the execution.  The sample $A_i$ for each item $i$ 
contains all elements $j$ such that
\begin{equation} \label{condprob:eq}
\frac{w_j}{r_j} (u \mid S)_{ij} \geq \tau\ .
\end{equation}
We also maintain
an estimate of  the  marginal influence $\Inf(i \mid S)$.  The
estimate is an inverse probability estimate  that is
computed from all elements satisfying \eqref{condprob:eq}.  The
probabilities are the inclusion probability of $j\in A_i$ over the random selection of
$r_j$:  The probability that an element $j$ satisfies \eqref{condprob:eq}
is $\min\{1,(u\mid S)_{ij}/\tau\}$ and its contribution to the marginal influence
of $i$ is $(u\mid S)_{ij}$.  The estimate is therefore
\begin{eqnarray} 
\widehat{\Inf}(i \mid S) &=& \sum_{j\in A_i} \frac{(u\mid
  S)_{ij}}{\min\{1,(u\mid S)_{ij}/\tau}\} \nonumber \\ &=& \sum_{j\in A_i} \max\{ (u
\mid S)_{ij} , \tau \}\ . \label{est:eq}
\end{eqnarray}

The threshold $\tau$ is decreased and samples and estimates are accordingly updated
until at least one item $i$ satisfies $\widehat{\Inf}(i \mid S) \geq k \tau$, for a parameter $k$. 
The item $\arg\max_i \widehat{\Inf}(i \mid S)$ with maximum 
estimate is then added to the seed set. Samples and estimates are then updated to reflect new marginal utility values.
 Note that any item with a sample of size
$\geq k$ has estimate $\geq k\tau$. Therefore, the total size of all
samples at any given point need not exceed $O(kn)$.
A choice of $k = O(\epsilon^{-2} \log n )$ would guarantee with high
probability that the selected item has marginal influence that is at
least $(1-\epsilon)$ of the maximum.  It is often possible to
adaptively determine the stopping condition on sample size and work
below this worst-case bound \cite{binaryinfluence:CIKM2014,timedinfluence:2015}.

 Efficient maintenance of the samples and estimates requires careful data
 structures and updates.   The algorithm maintains {\em inverted samples}:
$\Index{j}$ for an element $j$ contains all items such that $j$
was sampled for $i$.  The items are added using the reverse sorted
access oracle $\revSO[j]$ and kept in that order by decreasing $u_{ij}/r_j$.  Note that this
is the same order as decreasing $(u \mid S)_{ij}$ for any $S$. 
(Note that $r_j$ is fixed here, so the order of $u_{ij}$
stays the same).   The inverted sample of $j$, $\Index{j}$, is
logically partitioned into three segments.
The first segment of $\Index{j}$ contains {\em $H$ entries}, defined
as those with utility $(u\mid S)_{ij} \geq \tau$. 
The second contains {\em $M$ entries} which have utility 
$(u\mid S)_{ij} < \tau$ but satisfy \eqref{condprob:eq}, and the last includes {\em $L$ entries}.  $L$ entries 
are those that entered as $M$ or $H$, but no longer satisfy \eqref{condprob:eq}.
They are kept because they may become useful later as $\tau$
decreases.  The algorithm maintains indices $\HM{j}$ and $\ML{j}$ to
the position of each segment. It also maintains the sum $\est.H$ of $H$ entries
and the number $\est.M$ of $M$ entries, which suffice to compute the
estimate~\eqref{est:eq} as $\est.H+\tau \est.M$.

  As said, the sample data structures are modified when $\tau$ is decreased and
  when an item is added to the seed set.   The former results in
 new entries getting retrieved from the oracle and entries being
 ``upgraded'' from M/L to H and and from L to M/H.  This is done by
 the function \MoveUp.  The latter results
 in decreased marginal utility values $(u | S)_{ij}$ and in entries
 being ``downgraded'' and is handled by \MoveDown.

  To facilitate working with marginal utility $(u | S)_{ij}$ and
  computing it from  $u_{ij}$, the  algorithm maintains a {\em utility digest}
  $\digest{j}$ for
  each element $j$ (see Section \ref{digest:sec}).  The digest is initially empty and is kept current
  by the algorithm as seeds are added.  The digest support the forward
  search oracles.

  The running time (computation) bound analysis follows that of
  \cite{timedinfluence:2015} with some adaptations to our general
  aggregation functions and use of oracles.  

  We first bound the number of calls in forward searches.  An element $j$ can be returned by a
  forward search from $i$ when $u_{ij}$ is larger than the $\ell$th largest value in $U_{Sj}$.   
The total number of calls by forward searches is the sum over elements of the number of
  times top-$\ell$ set is modified as we add seed items.  
For analysis, we partition the greedy sequence into phases such that the
marginal influence decreases by at most $1-\epsilon$ in each phase.
Our approximate greedy sequence has the property that at each step we
  make a near uniform selection from all items with marginal influence
  that is at least $1-O(\epsilon)$ of the maximum.   Thus, the expected
  number of updates in each phase is at most $\ell \ln n'$, where $n'\leq n$
  is the number of items in the phase.  The algorithm terminates when
  the maximum marginal influence drops below $1/n^2$ of its initial
  value.  Therefore, the total number of phases is $O(\epsilon^{-1}
  \log n)$ and the total number of calls of the forward search oracle
  is $O(\ell n   \epsilon^{-1} \log^2 n)$.




\notinproc{
The algorithm simplifies for coverage problems with
 uniform utility to follow
 the basic SKIM design \cite{binaryinfluence:CIKM2014}.  In this case, marginal
 contributions are either $0$ or $1$ and the reverse
 sorted access oracles $\revSO[j]$ are invoked 
 sequentially by increasing $r_j$, removing the need to
 concurrently maintain active oracles for many elements.  The samples
 are uniform and do not need to be broken up to segments and
 estimates correspond to samples size.  The additional
 machinery was introduced in \cite{timedinfluence:2015} in order to handle
smooth distance-based utility functions.}

\begin{function}\caption{NextSeed()}
{\scriptsize
\KwOut{The item $i$ which maximizes $\est.H[i]+\tau \est.M[i]$, if 
  happy with estimate. 
Otherwise $\perp$.}
\SetKw{True}{true}
\SetKw{False}{false}
\While{\True}{
\lIf{max priority in $\Qcands < k\tau$}{\Return{$\perp$}}
\Else{
 Remove maximum priority $i$ from $\Qcands$\;
$\hat{I}_i \gets \est.H[i] + \tau \est.M[i]$\;
\If{$\hat{I}_i \geq k\tau$ and $\hat{I}_i \geq $  $\max$ in 
  $\Qcands$}{$I_i\gets \MargGain{i}$;\\ \lIf{$I_i \geq  (1-1/\sqrt{k})\hat{I}_i$}{\Return{$(i,\hat{I}_i)$}}\Else{Place $i$ with priority $I_i$ in 
    $\Qcands$\; \Return{$\perp$}}}
\Else{Place $i$ with priority $\hat{I}_i$ in $\Qcands$}
}}
} 
\end{function}

\begin{function}\caption{MoveUp()  Update estimates after
    decreasing $\tau$}
{\scriptsize
\DontPrintSemicolon
\ForEach{ $j$ in $\Qhml$ with priority $\geq \tau$}{ delete
  $j$ from $\Qhml$\;
\tcp{Process $\Index{j}$}
\If(\tcp*[f]{move entries from M/L to H}){$\HM{j}\not=\perp$}{
  \While{$\HM{j} < |\Index{j}|$ and 
    $i \gets \Index{j}[\HM{j}]$ satisfies 
    $(c\gets (u\mid S)_{ij} \geq \tau$}{
$\increase{\est.H[i]}{c}$ \\
 \If(\tcp*[f]{Entry was M}){$\ML{j} =\perp$ or $\ML{j} > \HM{j}$}{$\decrease{\est.M[i]}{1}$} 
  $\increase{\HM{j}}{1}$}
\If{$\ML{j}\not=\perp$ and $\ML{j} < \HM{j}$}{$\ML{j} \gets \HM{j}$}
\If{$\HM{j} \geq |\Index{j}|$}{$\HM{j}\gets\perp$;
  $\ML{j}\gets\perp$}
}
\If(\tcp*[f]{Move from L to M}){$\ML{j} \not= \perp$}
{
  \While{$\ML{j} < |\Index{j}|$ and 
    $i \gets \Index{j}[\ML{j}]$ satisfies 
    $(u \mid S)_{ij} \geq
    r_j\tau$}{$\increase{\ML{j}}{1}$; $\increase{\est.M[i]}{1}$}
\lIf{$\ML{j} \geq |\Index{j}|$}{$\ML{j}\gets\perp$}
}
$\UpdateReclassThresh{j}$ \tcp*{update $\Qhml$}
}
} 
\end{function}

\begin{function}\caption{UpdateReclassThresh()$(j)$}
{\scriptsize
\KwOut{Update priority of element $j$  in $\Qhml$}
$c\gets 0$;\\
\If{$\HM{j}\not=\perp$}{$i \gets \Index{j}[\HM{j}]$;
  $c\gets (u\mid S)_{ij}$}
\If{$\ML{j}\not=\perp$}{$i \gets \Index{j}[\ML{j}]$;
  $c\gets \max\{c, (u \mid S)_{ij}/r_j\}$}
\If{$c>0$}{update priority of $j$ in $\Qhml$ to $c$} 
}
\end{function}

\begin{function}[t]\caption{MoveDown() $(j ,x)$}
{\scriptsize
\KwOut{Update estimate contribution of element $j$ when item $h$ with utility $x=u_{hj}$ is
    added to $S$ to have $S^{+} = S \cup\{h\}$. }
\tcp{For each item $i\not\in S$, the marginal utility values $(u
    \mid S^+)_{ij} \leq (u  \mid S)_{ij}$ can be computed from
    $\digest{j}$ (computed for $S$) using $x=u_{hj}$ and $u_{ij}$:
 $(u \mid S)_{ij}
    \equiv \digest[j].marg(u_{ij})$   $(u \mid S^+)_{ij} \equiv
    \digest[j].{\text AddMarg}(u_{ih},u_{ij})$.}
\DontPrintSemicolon
$y \gets 0$; $t \gets \perp$; $\HM{j}\gets \perp$; \\ 
$z \gets
|\Index{j}|-1$; 
\lIf(\tcp*[f]{last non-L position in \Index{j}})
{$\ML{j} \not= \perp$}{$z \gets\ML{j}$}
$\ML{j}\gets
\perp$ \\
\While{$y \leq z$}{
$i \gets \Index{j}[y]$;\\
 \If(\tcp*[f]{entry was H}){$(u \mid S)_{ij} \geq \tau$}
 {$\decrease{\est.H[i]}{(u \mid S)_{ij} }$;\\ 
  \If(\tcp*[f]{is H}){$(u \mid S^+)_{ij} \geq
    \tau$}{$\increase{\est.H[i]}{(u \mid S^+)_{ij} }$}
 \ElseIf(\tcp*[f]{is M}){$(u \mid S^+)_{ij} \geq r_j 
    \tau$}{$\increase{\est.M[i]}{1}$; \lIf{$\HM{j}=\perp$}{$\HM{j}=y$}}
  \ElseIf(\tcp*[f]{truncate}){$(u \mid S^+)_{ij} =0$}{\lIf{$t=\perp$}{$t=y$}}
  \Else(\tcp*[f]{is L}){\lIf{$\ML{j}=\perp$}{$\ML{j}
        \gets y$}}
 }
\ElseIf(\tcp*[f]{entry was M}){($u \mid S)_{ij}  \geq r_j
  \tau$}
{
\If(\tcp*[f]{is M}){$(u \mid S^+)_{ij}  \geq r_j
  \tau$}{\lIf{$\HM{j}=\perp$}{$\HM{j} \gets y$}}
 \Else(\tcp*[f]{is not M}){$\decrease{\est.M[u]}{1}$;\\
  \eIf(\tcp*[f]{truncate}){$(u \mid S^+)_{ij}  = 0$}{\lIf{$t=\perp$}{$t=y$}}(\tcp*[f]{is
    L}){\lIf{$\ML{j}=\perp$}{$\ML{j}\gets y$}}
}}
Priority of $i$ in $\Qcands$ $\gets \est.H[i]+\tau \est.M[i]$
\tcp*[l]{Can be lazy}
 $\increase{y}{1}$
}
\lIf{$t\not=\perp$}{truncate $\Index{j}$ from $t$ on.}
\Else(\tcp*[f]{clean tail}){$t\gets |\Index{j}|-1$; \\
\While{$t \geq 0$ and $i \gets \Index{j}[t]$
    has $(u \mid S^+)_{ij} = 0$}{$\decrease{t}{1}$}
truncate $\Index{j}$ at position $t+1$ on}
Remove element $j$ from $\Qhml$\;
$\UpdateReclassThresh{j}$ \tcp*{Update $\Qhml$}
} 
\end{function}

  \small
\notinproc{\bibliographystyle{plain}}
\onlyinproc{\bibliographystyle{IEEEtran}}
\bibliography{cycle}

\end{document}